\documentclass[final,5p,times,twocolumn]{elsarticle}
\usepackage{graphicx}      
\usepackage{natbib}        
\usepackage{subfigure}
\usepackage{enumerate}
\usepackage{appendix}
\usepackage{amsthm,amsmath,amssymb}
\usepackage{lipsum}
\usepackage{cuted}
\usepackage{mathrsfs}
\usepackage[table,xcdraw]{xcolor}
\usepackage{float}
\usepackage{stfloats}
\usepackage{enumitem}
\usepackage{hyperref}
\usepackage[switch]{lineno}
\usepackage{mhchem}
\usepackage{bbding}
\usepackage{bbm}
\usepackage{arydshln}
\usepackage{nomencl}
\usepackage{etoolbox}
\usepackage{nicematrix}
\usepackage{tikz}
\makenomenclature
\renewcommand*\nompreamble{\begin{multicols}{2}}
\renewcommand*\nompostamble{\end{multicols}}
\renewcommand\nomgroup[1]
{%
	\item[\bfseries
	\ifstrequal{#1}{S}{Symbols}{%
		\ifstrequal{#1}{M}{General Mathematical Symbols}{%
			 \ifstrequal{#1}{P}{Abbreviations}{}}}%
]}

\usepackage{framed}             
\usepackage{multicol}           
\usepackage{booktabs}
\usepackage{tabularx}
\usepackage[ruled,linesnumbered]{algorithm2e}

\newcommand{\RM}[1]{\mathrm{#1}}

\def\d{\mathrm{d}}

\def\dw{\d \omega}
\def\ejw{\RM{e}^{\RM{j}\omega}}

\newtheorem{thm}{Theorem}
\newtheorem{cor}[thm]{Corollary}

\newtheorem{assum}[thm]{Assumption}
\newtheorem{prop}[thm]{Proposition}

 \theoremstyle{definition}

\newtheorem{rem}[thm]{Remark}

\SetKwInput{KwInput}{Input}                
\SetKwInput{KwOutput}{Output}              

\begin{document}

\begin{frontmatter}

\title{On the optimality of Kalman Filter for Fault Detection}


\author[zju]{Jinming Zhou}
\ead{zhoujinming@zju.edu.cn}

\author[zju]{Yucai Zhu\corref{mycorrespondingauthor}}
\cortext[mycorrespondingauthor]{Corresponding author}
\ead{zhuyucai@zju.edu.cn}

\address[zju]{State Key Laboratory of Industrial Control Technology, College of Control Science and Engineering, Zhejiang University, Hangzhou 310027, China}

\begin{abstract}                
    Kalman filter is widely used for residual generation in fault detection. It leads to optimality in fault detection using some performance indices and also leads to statistically sound residual evaluation and threshold setting. This paper shows that these nice features do not necessarily imply an optimal fault detection performance. Based on a performance index related to fault detection rate and false alarm rate, several occasions where Kalman filter should not be used are pointed out; further the residual evaluation and threshold setting are discussed, in which it is pointed out that in stochastic setting an optimal statistical test of Kamlan filter is not related to optimality of commonly used detection performance indicators. The theoretical analysis is verified through Monte Carlo simulations and Tennessee Eastman process (TEP) dataset.
 \end{abstract}
 

\end{frontmatter}


\section{Introduction}
\label{sec:intro}
Fault detection, also known as change detection, is a very important and well researched topic in many engineering fields e.g. process industry, mechanical systems and biomedicine \citep{basseville1993detection}. Its goal is to detect the abnormal events in system and related signals. Among numerous branches of this field, model-based method is a natural and classical one, of which the fundamental is a mathematical model describing the relationships between signals. The essentials of model-based methods contain three parts: 1) residual generation; 2) residual evaluation and 3) threshold setting \citep{ding2008model}. 

There are two different settings when considering residual generation: deterministic \citep{ding2008model,chen2012robust} and stochastic \citep{basseville1993detection,gustafsson2000adaptive}. In the stochastic setting, there are numerous papers choosing Kalman filter (KF) as the basic residual, see  \cite{simani2008fault,wei2010sensor,dong2012robust,zhang2018adaptive}. The choice is mainly due to KF provides an optimal state estimation, and the white innovation often leads a statistically sound threshold setting. Further, its natural link to parameter estimation can be utilized to design adaptive methods for time-varying system \cite{xu2004residual,zhang2018adaptive}. In frequency domain, the disturbance in a KF-based residual has a flat spectrum, which is same as the unified solution derived in the deterministic setup \cite{ding2000unified}. The unified solution is optimal concerning some system norm based indices and it solves the optimal integrated design problem \cite{ding2008model}. The simplicity of KF-based residual and the unified solution seemingly suggests KF would be a good choice for FD.

The three most commonly used indicators of FD performance are  fault detection rate (FDR), false alarm rate (FAR) and mean time to detection (MT2D). The question is that, does KF-based residual really leads optimality in these indicators? This paper explores the question in a stochastic setting. The discussions focus on two aspect: performance comparison of KF-based residual to others using an index related to FDR and FAR; statistical test issue (including residual evaluation and threshold setting). Utilizing the equivalence between KF-based residual and prediction error, and the optimal filter framework developed in \cite{zhou2021identification}, several occasions where Kalman filter is not optimal will be pointed out. In this procedure, some additional implementation details of the optimal filter scheme are also given.

The paper unfolds as follows: Section~\ref{sec:pre} gives the necessary preliminaries and states the equivalence between KF-based residual and prediction error. The main results are contained in Section~\ref{sec:main}. Section~\ref{sec:sim} uses Monte Carlo simulations to validate the analysis. Section~\ref{sec:TEP} reports the validation in TEP. Section~\ref{sec:con} is the conclusion.

\section{Preliminaries}
\label{sec:pre}
The paper considers a stochastic linear time-invariant (LTI) system, and its state-space form is
\begin{subequations}
   \begin{align}
      x(t+1) &= Ax(t)+Bu(t)+w(t), \\
   y(t)   &= Cx(t)+\nu (t) 
   \end{align}
   \label{eq:sys_ss}
\end{subequations} 
where $w(t)$ and $\nu (t)$ are assumed to be Gaussian white noise. Denote $\mathbb{E}w(t)w^{\top}(t):=\varSigma_{w}$,  $\mathbb{E}\nu(t)\nu^{\top}(t):=\varSigma_{\nu}$ and $\mathbb{E}w(t)\nu^{\top}(t):=\varSigma_{w\nu}$. A KF-based residual takes the form:
\begin{subequations}
   \begin{align}
      \hat{x}(t+1)&=A\hat{x}(t)+Bu(t)+K\left[y(t)-C\hat{x}(t)\right], \\
      \varepsilon (t)&=y(t)-C\hat{x}(t)
   \end{align}
   \label{eq:kf_res}
\end{subequations}
where $K$ denotes the Kalman gain, which can be obtained by solving the algebraic Riccati equation.
Denote the covariance of $\varepsilon(t)$ as $\mathbb{E}\varepsilon(t)\varepsilon^{\top}(t):=\varLambda$.
Alternatively, the following transfer-function model can be considered:
\begin{subequations}
   \begin{align}
      y(t)&=G(q)u(t)+H(q)\varepsilon(t) \\
      G(q)&=C(qI-A)^{-1}B \\
      H(q)&=C(qI-A)^{-1}K+I 
   \end{align}
   \label{eq:sys_tf}
\end{subequations}
where $q$ denotes the forward time shifter. Based on the discussion in Chapter~4.3 in \cite{ljung1999system} and with some algebraic operations, the following proposition holds.
\begin{prop}
   System~(\ref{eq:sys_ss}) and System~(\ref{eq:sys_tf}) are equivalent, further $\varepsilon(t)$ in (\ref{eq:kf_res}) can alternatively be written as:
   \begin{equation}
      \varepsilon(t)=H^{-1}(q)\left[y(t)-G(q)u(t)\right].
      \label{eq:pe}
   \end{equation}
\end{prop}

This proposition implies the KF-based residual is equivalent to (\ref{eq:pe}), also known as prediction error (PE) in system identification. The transfer-function model (\ref{eq:sys_tf}) and equation~(\ref{eq:pe}) will be used in subsequent analysis. In practice, fault can happen at input, output or state signals. However it is more convenient to work with a final term at output,
\begin{equation}
   y(t)=G(q)u(t)+v(t)+f(t)
\end{equation}
where $v(t):=H(q)e(t)$. As a competitor of PE, the output error (OE) that is equivalent to the residual delivered by a zero-gain observer can be written as
\begin{equation}
   \zeta(t)=y(t)-G(q)u(t).
   \label{eq:oe}
\end{equation}

For simplicity only single-input single-output (SISO) system without model error is considered. The results can however be extended to multi-input multi-output (MIMO) system with model error, by decomposing the system to multi-input single-output (MISO) subsystems and using a generalized disturbance term $\bar{v}$ composed of the total effects of noise and model error. See \cite{zhou2021identification} for details. 

\section{The optimality of KF-based residual}
\label{sec:main}
In this section, the optimality of KF-based residual will first be investigated through a performance index then in the residual evaluation and threshold setting issues. Several situations that KF-based residual does not outperform others will be pointed out.
\subsection{Residual comparison using a performance index}
Consider a general residual  
\begin{equation}
   r(t) = \bar{f}(t) + \bar{d}(t)
\end{equation}
where $\bar{f}$ and $\bar{d}$ denote the final effects of faults and disturbances in the residual. Take PE for instance, $\bar{f}(t)=H^{-1}(q)f(t)$, $\bar{d}(t)=\varepsilon(t)$. Introduce the following performance index:

\begin{subequations}
   \begin{align}
      & \mathcal{J}:={\mathbb{E}\{\left[{r}^{\RM{f}}(t)\right]^2\}}/{\mathbb{E}\{\left[{r}^{\RM{n}}(t)\right]^2\}} \\
      &=\frac{\lim_{N_{\RM{f}}\to\infty}\sum_{t=1}^{N_{\RM{f}}}\left[{r}^{\RM{f}}(t)\right]^2}{\lim_{N_{\RM{n}}\to\infty}\sum_{t=1}^{N_{\RM{n}}}\left[{r}^{\RM{n}}(t)\right]^2}
      =\frac{\Vert {r}^{\RM{f}}(t)\Vert^2_{\RM{pow}}}{\Vert {r}^{\RM{n}}(t)\Vert^2_{\RM{pow}}} \label{eq:J_time}\\ 
      &=\frac{\int_{-\pi}^{\pi}\varPhi_{\bar{f}}(\omega)\dw}{\int_{-\pi}^{\pi}\varPhi_{\bar{d}}(\omega)\dw}+1 =\frac{\Vert R_{\bar{f}}(\ejw) \Vert_2^2}{\Vert R_{\bar{d}}(\ejw)\Vert_2^2}+1\label{eq:J_freq}
   \end{align}
   \label{eq:J}
\end{subequations}

The subscripts f and n denote faulty and normal respectively. $\varPhi_{\cdot}$ denotes power spectrum, $\Vert\cdot\Vert_{\RM{pow}}$ denotes power norm, also called root mean square value,  $\Vert\cdot\Vert_{2}$ denotes $\mathcal{H}_2$ norm of a transfer function. $R_{\bar{f}}(\ejw)$ and $R_{\bar{d}}(\ejw)$ are spectral factor of the corresponding spectra. See Chapter 2.3 of \cite{ljung1999system} for detailed discussions about power spectrum and spectral factor. Notice that the first equal sign in (\ref{eq:J_freq}) is due to Parseval's relation.

A deterministic version of this index was first proposed in \cite{simani2008fault}, while in \cite{zhou2021identification} the authors extended it to stochastic setting. 
The index $\mathcal{J}$ measures the fault-to-noise ratio (FNR) from $\bar{f}$ to $\bar{d}$. 
A larger $\mathcal{J}$ means a higher detection performance. It is also closely related to the FDR and FAR, whose expressions will be given later (c.f. (\ref{eq:index})). For a given threshold, a larger power of faulty residual means a higher FDR while a smaller power of normal residual means a lower FAR. It should be emphasized that the information of fault is naturally embedded in $\mathcal{J}$.

Consider a specific case that is often encountered in practice, which is formalized below:
\begin{assum}
   The disturbance signal $v(t)$, the fault signal $f(t)$ and the generalized disturbance are all low-pass, moreover $\varPhi_{f}$ decays faster than $\varPhi_{v}$. 
\end{assum}

The power spectrum of a low-pass signal mentioned above mainly concentrates at the low frequencies and decays as frequency increases. The spectrum of $f(t)$ decays faster such that it is `more low-pass'. This assumption encompasses typical fault types like step, drift and stair signals. \citep{zhang1994early,isermann2011fault}.

\begin{thm}
   Under Assumption 1, the performance index in (\ref{eq:J}) of the output error (OE) residual is larger than the one of prediction error (PE) residual.
 \end{thm}

 The proof of this theorem can be found in \cite{zhou2021identification}. Based on Proposition~1 and this theorem, the following corollary is straightforward:
 \begin{cor}
   Under Assumption 1, the KF-based residual (\ref{eq:kf_res}) has a smaller performance index than that of the output error residual (\ref{eq:oe}).
 \end{cor}

In Chapter 7.9 and 12.4.3 of \cite{ding2008model}, the author proved that KF-based residual solves the $\mathcal{H}_i/\mathcal{H}_{\infty}$ optimization problem concerning fault to residual and disturbance to residual transfer matrices. Here the optimization object is changed so KF-based residual is no longer optimal. The main differences of $\mathcal{J}$ to the $\mathcal{H}_i/\mathcal{H}_{\infty}$ used in \cite{ding2008model} are that $\mathcal{J}$ is $\mathcal{H}_2$ related and involves the fault information.

For the general residual $r(t)$, a post-filter $Q(q)$ can be designed to achieve the maximum $\mathcal{J}$. If the performance index of $Q(q)r(t)$ reaches the maximum, $Q(q)$ is called optimal filter for $r(t)$. The following result gives the solution of the general residual:
 \begin{thm}
   Suppose that the order of $Q$ ($n_Q$) is allowed to tend to $\infty$, the optimal filter is a frequency selector satisfying
   \begin{equation}
     \int_{-\pi}^{\pi}|Q(\mathrm{e}^{\mathrm{j}\omega})|^{2}\varPhi_{z}(\omega)\mathrm{d}\omega =
     2\pi|Q(\mathrm{e}^{\mathrm{j}\omega_{0}})|^{2}\varPhi_{z}(\omega_{0}),
     \label{eq:Qopt}
   \end{equation}
   where
   \begin{equation}
     \omega_{0}: = \mathop{\arg\max}_{w} {\varPhi_{\bar{f}}(\omega)}/{\varPhi_{\bar{d}}(\omega)}
     \label{eq:wopt}
   \end{equation}
   and $\varPhi_{z}$ is the spectrum of an arbitrary quasi-stationary stochastic process $\{z(t)\}$. The maximum value of $J$ is
   \begin{equation}
     J_{\text{opt}} = {\varPhi_{\bar{f}}(\omega_{0})}/{\varPhi_{\bar{d}}(\omega_{0})}+1.
     \label{eq:wopt2}
   \end{equation} 
\end{thm}

The theorem is a generalization of Theorem~2 in \cite{zhou2021identification} where $r(t)$ is confined to OE only. The proof is similar hence it is omitted here. The interpretation of this theorem is that the best detection performance concerning one specific fault is achieved, when the frequency leading the largest FNR is selected.  
\begin{rem}
   When no model error exists, the condition that KF-based residual, or prediction error leads to the best performance is when $\varPhi_{v}\equiv \varPhi_f$ at every frequency, which can hardly be met in practice. In that case, the optimal filters is an all-pass filter and OE, PE, KF-based residual are all optimal. If it is not the case, one can always magnify the effect of fault using the optimal filter.
\end{rem}
\begin{rem}
   Notice that PE is obtained by filtering OE with $H^{-1}(q)$. This fact reveals that ideally (when $n_Q\rightarrow \infty$), letting $r(t)$ be OE or PE leads to the same optimal filters and same values of $\mathcal{J}$, because $H^{-1}(\ejw)$ is cancelled in (\ref{eq:wopt}) and (\ref{eq:wopt2}).
\end{rem}

According to Theorem~3 and Remark~2, there is no difference choosing whether OE or PE as the basic residual then using an optimal filter. This is true for the ideal case, but in practice one must use finite-order low-pass or band-pass filters as an approximation of the optimal frequency selector, in which case choice of basic residual $r(t)$ do affects the detection performance. Consider again the situation in Assumption~1, where the ideal optimal filter should select only the zero frequency. In practice, low-pass filters are used as approximations. Then the following result holds:
\begin{prop}
   Suppose that an ideal low-pass filter $Q(q)$ whose passband is $\left[0,\bar{\varOmega}\right]$ is used to both OE and PE. Denote their new performance indices as $\mathcal{J}_{\RM{oef}}$ and $\mathcal{J}_{\RM{pef}}$, then under Assumption~1, it holds that:
   \begin{equation}
      \mathcal{J}_{\RM{oef}} > \mathcal{J}_{\RM{pef}}.
   \end{equation}
\end{prop}

\begin{proof}
   First note that 
   \begin{equation}
      \frac{\int_{-\pi}^{\pi}\vert Q(\ejw)\vert^2 \varPhi^{\RM{f}}_{r}(\omega)\dw}{\int_{-\pi}^{\pi}\vert Q(\ejw)\vert^2\varPhi^{\RM{n}}_{r}(\omega)\dw}
      = \frac{\int_{0}^{\bar{\varOmega}}\varPhi^{\RM{f}}_{r}(\omega)\dw}{\int_{0}^{\bar{\varOmega}}\varPhi^{\RM{n}}_{r}(\omega)\dw},
   \end{equation}
   hence to compare $\mathcal{J}_{\RM{oef}}$ and $\mathcal{J}_{\RM{pef}}$, it is sufficient to compare:
   \begin{equation}
      \frac{\int_{0}^{\bar{\varOmega}}\varPhi_{f}(\omega)\dw}{\int_{0}^{\bar{\varOmega}}\varPhi_{v}(\omega)\dw}\text{ and }\frac{\int_{0}^{\bar{\varOmega}}\vert H(\ejw)\vert^{-2}\varPhi_{f}(\omega)\dw}{\int_{0}^{\bar{\varOmega}}\vert H(\ejw)\vert^{-2}\varPhi_{v}(\omega)\dw}.
   \end{equation}
   The following inequality can be verified:
   \begin{equation}
      \frac{\int_{0}^{\bar{\varOmega}}\vert H(\ejw)\vert^{-2}\varPhi_{v}(\omega)\dw}{\int_{0}^{\bar{\varOmega}}\varPhi_{v}(\omega)\dw}>\frac{\int_{0}^{\bar{\varOmega}}\vert H(\ejw)\vert^{-2}\varPhi_{f}(\omega)\dw}{\int_{0}^{\bar{\varOmega}}\varPhi_{f}(\omega)\dw}
   \end{equation}
   because at frequency interval $\left[0,\bar{\varOmega}\right]$, $f(t)$ is more low-pass than $v(t)$, after being filtered by the high-pass filter $H^{-1}(q)$ its relative energy loss is larger. Thus
   \begin{equation}
      \frac{\int_{0}^{\bar{\varOmega}}\varPhi_{f}(\omega)\dw}{\int_{0}^{\bar{\varOmega}}\varPhi_{v}(\omega)\dw}>\frac{\int_{0}^{\bar{\varOmega}}\vert H(\ejw)\vert^{-2}\varPhi_{f}(\omega)\dw}{\int_{0}^{\bar{\varOmega}}\vert H(\ejw)\vert^{-2}\varPhi_{v}(\omega)\dw},
   \end{equation}
   the proposition has been proved.
\end{proof}

In practice, the bandwidth of $Q$ must be carefully designed considering both the realization issue and the detection speed. The narrower the bandwidth, the closer to the optimal performance index, but at the same time, a lower detection speed \citep{zhou2021identification}. It should be remarked that, although Proposition~4 points out a situation that choosing PE or KF-based residual as $r(t)$ is worse than OE, there may exist situations where they are better.

\subsection{Discussions on statistical test}
In the stochastic setup, the residual evaluation and threshold setting can altogether be called statistical test. A model-based fault detection is typically based on the following decision logic:
\begin{equation}
   J\left(r(t)\right) < J_{\RM{th}} ~\Rightarrow ~ \text{fault free, otherwise faulty.}
   \label{eq:des_rule}
\end{equation}
$J\left(r(t)\right)$ is the evaluation function of $r(t)$ while $J_{\RM{th}}$ is the threshold. Based on $J$ and $J_{\RM{th}}$, the three performance indicators mentioned in Section~\ref{sec:intro} can be defined:
\begin{equation}
   \begin{aligned}
   &\RM{FDR} = \Pr \left(J_{T^{2}} < J_{T^2,\RM{th}}|{f}\neq 0\right), \\
   &\RM{FAR} = \Pr \left(J_{T^{2}} < J_{T^2,\RM{th}}|{f}= 0\right), \\
   &\RM{MT2D}= \Sigma_{k=0}^{k_{\RM{stop}}}k\cdot\Pr(k|{f}\neq 0)
  \end{aligned}
  \label{eq:index}
\end{equation}
where $\Pr(k|{f}\neq 0)$ is the probability that the FD system alarms for the first time at $k\in\left[0,k_{\RM{stop}}\right]$ when faults exist.

The KF-based residual is white and normal. Other aforementioned residuals are typically colored (correlated), making the statistical test difficult. According to Neyman-Pearson lemma and Generalized Likelihood Test (GLR), the following evaluation function is popular for KF-based residual:
\begin{equation}
   J_{\RM{KF}}(t):=\sum_{k = 0}^{s}\varepsilon^{\top}(t-k)\varLambda ^{-1}\varepsilon(t-k),
   \label{eq:J_KF}
\end{equation} 
with the threshold
\begin{equation}
   \Pr \left(\chi^2((s+1)\RM{dim}\varepsilon)>J_{\RM{KF,th}}\right)=\alpha
\end{equation}
where $\chi^2$ denotes $\chi^2$ distribution. Notice that $s$ is an adjustable parameter. Using the above setting, $\RM{FAR}=\alpha$ and given a FAR the FDR can be maximized. However this optimality only holds specifically for statistical tests of KF-based residual. In a real FD task the detection performance also counts. If FNR in a residual is poor, the detection will fail despite the statistical optimality.

In \cite{ding2008model,ding2021advanced}, an integrated design problem is considered which aims at developing an optimal post-filter to maximize FDR given a FAR. The author defines FAR and FDR of deterministic system with bounded disturbance, then uses these concepts to prove that the unified solution is the optimal solution of the integrated design problem. Later the author claims KF-based residual is also optimal for stochastic system using the fact that both two residuals have flat spectra. Notice that in above procedure, the optimality of KF-residual concerning (\ref{eq:index}) is actually not proved. 

It is worthy to discuss the well-known $T^2$ statistics, which is used as evaluation function of OE and its filtered version in \cite{zhou2021identification}. The $T^2$ statistic of a residual is defined as:
\begin{equation}
   J_{T^{2}}\left({r}(t)\right) = [{r}(t)-\mu]^{\top} {S}^{-1} [{r}(t)-\mu],
   \label{eq:T2}
\end{equation}
where $\mu$ and $S$ denote the mean and covariance estimated from fault-free residual $r(t)$. The corresponding threshold is
\begin{equation}
   \Pr \left(\frac{p(N^{2}-1)}{N(N-p)}F(p,N-p)>J_{T^2\RM{,th}}\right)=\alpha.
   \label{eq:threshold}
 \end{equation} 
where $N$ is the sample length of fault-free data, $p=\RM{dim}(r)$, and $F$ denotes $F$ distribution. For KF-based residual and PE, $T^2$ statistic serves as an alternative when the residual covariance is not known exactly. For large $N$, a FAR nearly equaling $\alpha$ and an optimal statistical test can still be obtained. But for those colored residuals, e.g. OE, filtered OE and filtered PE, there will be bias between FAR and $\alpha$, and the statistical test is not optimal. 

\begin{table*}[!t]\scriptsize
    \setlength{\abovecaptionskip}{-0.5cm}
    \setlength{\belowcaptionskip}{-0cm}
	\caption{Performance indicators of the Monte Carlo simulations in Case 1: step fault.}
	\centering
   \begin{tabularx}{0.7\textwidth}{Xlllllll}
      \toprule
            & $J_{\RM{KF}}$, $s=10$  &  $J_{\RM{KF}}$, $s=100$ & $J_{\RM{KF}}$, $s=2000$ & $J_{T^{2}}$, OE    & $J_{T^{2}}$, PE    & $J_{T^{2}}$, OEF   & $J_{T^{2}}$, PEF \\
      \midrule
      FDR (\%) & 2.1470  & 6.4026  & 75.0384  & 79.6481  & 1.4282  & \textcolor[rgb]{ .753,  0,  0}{91.9657 } & 81.4036  \\
      FAR (\%) & 0.9983  & 0.9846  & 0.9646  & 0.8982  & 0.9994  & 1.1045  & 1.0004  \\
      MT2D (sample) & \textcolor[rgb]{ .682,  .667,  .667}{1.0000 } & \textcolor[rgb]{ .682,  .667,  .667}{1.0000 } & \textcolor[rgb]{ .753,  0,  0}{1.0000 } & 12.8102  & \textcolor[rgb]{ .682,  .667,  .667}{1.0000 } & 158.5254  & 112.1740  \\
      \bottomrule
      \end{tabularx}%
	\label{tab:step}
\end{table*}

\begin{table*}[!t]\scriptsize
    \setlength{\abovecaptionskip}{-0.1cm}
    \setlength{\belowcaptionskip}{-0cm}
	\caption{Performance indicators of the Monte Carlo simulations in Case 2: drift fault.}
	\centering
   \begin{tabularx}{0.7\textwidth}{Xlllllll}
      \toprule
      & $J_{\RM{KF}}$, $s=10$  &  $J_{\RM{KF}}$, $s=100$ & $J_{\RM{KF}}$, $s=2000$ & $J_{T^{2}}$, OE    & $J_{T^{2}}$, PE    & $J_{T^{2}}$, OEF   & $J_{T^{2}}$, PEF \\
      \midrule
      FDR (\%) & 2.5153  & 10.3384  & 36.6471  & 61.4515  & 1.4568  & \textcolor[rgb]{ .753,  0,  0}{70.8666 } & 70.4854  \\
      FAR (\%) & 0.9989  & 1.0092  & 0.9603  & 1.0177  & 1.0036  & 1.0727  & 0.9793  \\
      MT2D (sample) & \textcolor[rgb]{ .682,  .667,  .667}{375.0896 } & \textcolor[rgb]{ .682,  .667,  .667}{1339.2708 } & 2878.1704  & \textcolor[rgb]{ .753,  0,  0}{846.7202 } & \textcolor[rgb]{ .682,  .667,  .667}{100.5828 } & 1183.4510  & 1159.8462  \\
      \bottomrule
      \end{tabularx}%
	\label{tab:drift}
\end{table*}

\begin{table*}[!t]\scriptsize
    \setlength{\abovecaptionskip}{-0.1cm}
    \setlength{\belowcaptionskip}{-0cm}
	\caption{Performance indicators of the Monte Carlo simulations in Case 3: sine fault.}
	\centering
   \begin{tabularx}{0.7\textwidth}{Xlllllll}
      \toprule
      & $J_{\RM{KF}}$, $s=10$  &  $J_{\RM{KF}}$, $s=100$ & $J_{\RM{KF}}$, $s=2000$ & $J_{T^{2}}$, OE    & $J_{T^{2}}$, PE    & $J_{T^{2}}$, OEF   & $J_{T^{2}}$, PEF \\
      \midrule
      FDR (\%) & 3.5908  & 16.8097  & \textcolor[rgb]{ .753,  0,  0}{85.1773 } & 1.0163  & 1.8070  & 78.5751  & 78.5780  \\
      FAR (\%) & 1.0010  & 1.0182  & 0.9468  & 1.0491  & 1.0070  & 1.1193  & 1.0641  \\
      MT2D (sample) & \textcolor[rgb]{ .682,  .667,  .667}{95.7858 } & \textcolor[rgb]{ .682,  .667,  .667}{203.7510 } & 691.9424  & \textcolor[rgb]{ .682,  .667,  .667}{1662.3111}  & \textcolor[rgb]{ .682,  .667,  .667}{37.3756 } & \textcolor[rgb]{ .753,  0,  0}{271.2402 } & 272.0680  \\
      \bottomrule
      \end{tabularx}%
	\label{tab:sine}
\end{table*}

\section{Simulation study}
\label{sec:sim}
\subsection{System description}
Consider a SISO system in the form of (\ref{eq:sys_ss}). The system matrices are given as follows:
\begin{equation}
   \begin{aligned}
      &A=\left(\begin{array}{ccc}
         0 & 1\\
         -0.9063 & 1.905 \\
    \end{array}\right),~
      B = \left(\begin{array}{ccc}
         1\\
         1 \\
    \end{array}\right),~
    C=\left(1~0\right),\\
    &\varSigma_w=0.01I_2,~ \varSigma_{\nu}=0.01,~ \varSigma_{w\nu}=0.
   \end{aligned}
\end{equation}
Under the assumption that no model error exists, the input has no effect on the residual hence it is set zero for simplicity.

\subsection{Monte Carlo simulation}
In this section, different residuals and evaluation functions will be tested, including: KF-based residual evaluated by $J_{\RM{KF}}$ (\ref{eq:J_KF}) with different $s$; OE, PE and their filtered version evaluated by $J_{T^{2}}$ (\ref{eq:T2}). 5000 simulations will be run and the FDRs, FARs and MT2Ds will be compared. In each simulation, $N=10000$ samples are generated for $T^2$ threshold setting, and $N_{\RM{FD}}=10000+s$ samples are generated to test the FD performance, in which faults are assumed to happen at $N_f=5000+s$. In calculation of the evaluation functions, $J_{T^{2}}$s use the samples from $s+1$ to end while $J_{\RM{KF}}$s use all samples. In this way all methods utilize equal amount of data fo FD. Concerning the threshold setting, for those white residual $\alpha$ is set to $0.99$ and for those colored (OE, filtered OE and filtered PE) $\alpha$ is set to $0.993$. This makes FARs of all the methods nearly the same, which guarantees fair comparisons of FDRs.

Notice that due to large $N$ (normal data length), evaluating PE by $J_{T^{2}}$ is equivalent to evaluating it by $J_{\RM{KF}}$ with $s=0$. In the tables given below, the method that gives highest FDR is marked red; the MT2Ds of the methods that give FDR $<20\%$ make little sense so they are marked grey, the lowest MT2D apart from the grey ones is marked red.  

\begin{figure*}[!t]
    \vspace{-25pt}
    \setlength{\abovecaptionskip}{-0.2cm}
	\setlength{\belowcaptionskip}{-0.4cm}
   \begin{center}
   {\includegraphics[width=0.7\textwidth]{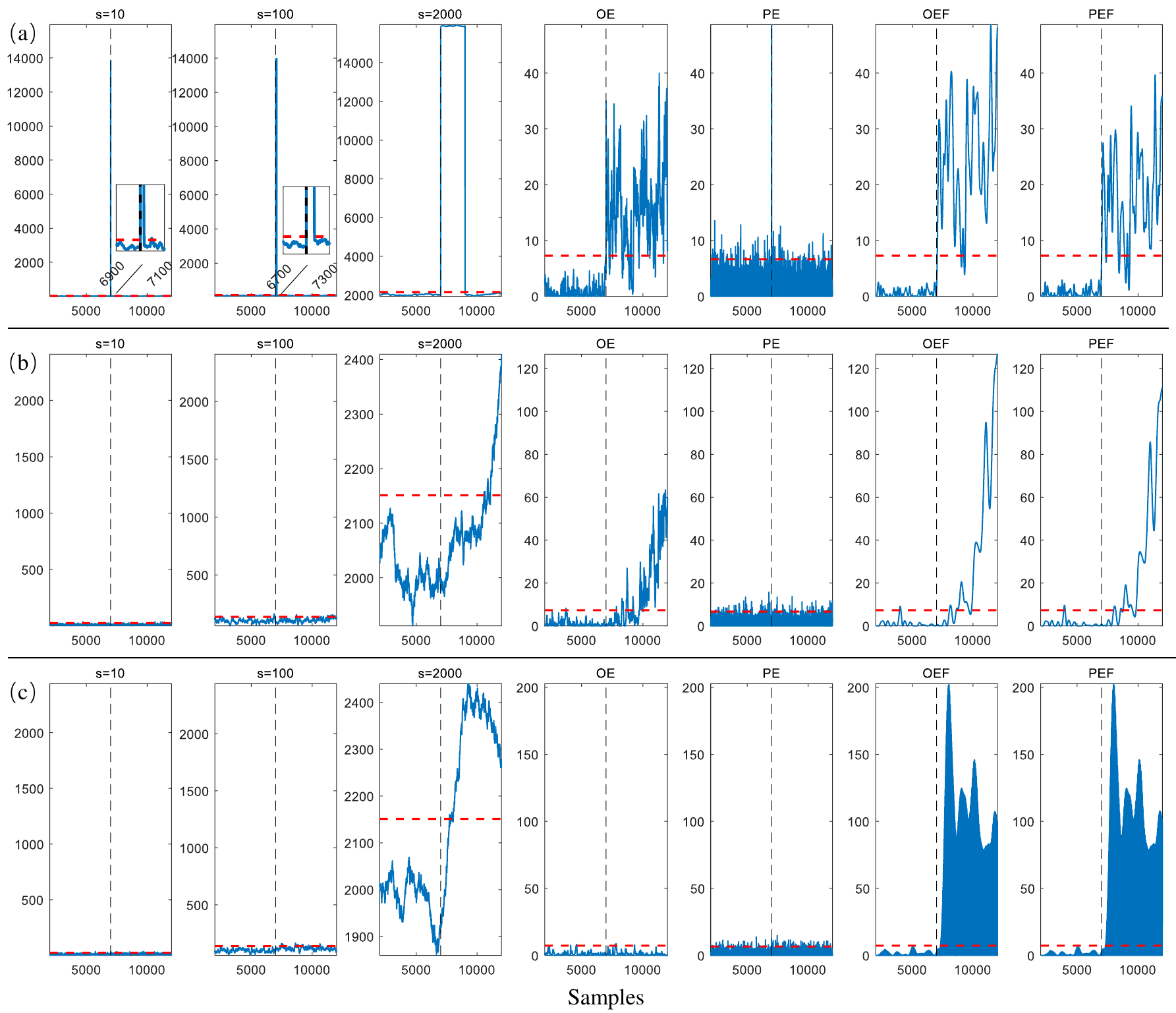}}
   \end{center}
   \caption{Different methods in one realization. (a) Case 1: step fault. (b) Case 2: drift fault. (c) Case 3: sine fault.}
   \label{fig:comp}
\end{figure*}

\textbf{Case~1: step fault}

 In this case, the fault is a step signal with amplitude 30; and a Butterworth low-pass filter with 0.02 rad/sample cutoff frequency is used as the optimal filter. The results are shown in Table~\ref{tab:step} and one randomly chosen realization is shown in Fig.~\ref{fig:comp}(a). The results imply that OE+filter (OEF) leads to the best performance while PE and $J_{\RM{KF}}$ with small $s$ deliver very poor performance. All methods based on pure KF-based residual or PE result in a large peak after fault happens, but decrease sharply afterwards (the peak width depends on $s$). The large peak leads the smallest MT2D but this is undesirable as the evaluation function varies inconsistently with the actual fault. After combing PE with optimal filter (PEF), FDR is greatly enhanced but still worse than OEF, which confirms the declaration in Proposition~4.

 \textbf{Case~2: drift fault}

In this case, fault is a drift signal with amplitude 60. The cutoff frequency of the low-pass filter is set to 0.01 rad/sample. The results are shown in Table~\ref{tab:drift}, and Fig.~\ref{fig:comp}(b). OEF remains best while the gap between OEF and PEF is very small. This is due to the bandwidth of the filter is narrower than the one in Case 1, which is closer to the ideal frequency selector, using which OEF and PEF deliver the same performance (see Remark~2). All methods based on pure KF-based residual or PE result in poor FDR and MT2D. 

\textbf{Case~3: sine fault}

In this case, fault signal is a sine signal with amplitude 0.6 and frequency 0.4. Optimal filter is chosen as a Butterworth band-pass filter with passband $[0.395,0.405]$ rad/sample. The results are shown in Table~\ref{tab:sine}, and Fig.~\ref{fig:comp}(c). The results show that $J_{\RM{KF}}$ with $s=2000$ gives the highest FDR but a large MT2D; meanwhile OEF and PEF give slightly small FDR but much smaller MT2D. 

\textbf{Summary}

Based on the results above, the following summaries can be made: 1) the optimal filter scheme can offer fairly high and stable detection performance, although $T^2$ test is not optimal for such colored residuals; 2) pure KF-residual cannot offer as high performance as optimal filter. By increasing $s$ of $J_{\RM{KF}}$, its FDR can be raised but a large $s$ leads a slow detection speed and poses demanding requirement for the amount of data. 

\begin{table*}[!htbp]\scriptsize
    \vspace{-15pt}
    \setlength{\abovecaptionskip}{-0.3cm}
	\setlength{\belowcaptionskip}{-0cm}
   \centering
   \caption{FDRs of different methods of  fault 3, 9, 15.}
     \begin{tabularx}{0.6\textwidth}{llllllll}
     \toprule
     Fault & SCVA  & SAP   & OE    & OEFB  & KF ($s=$0) & KF ($s=$10) & KF ($s=$20) \\
     \midrule
     3     & 16.3  & 6.4   & 18.6  & \textcolor[rgb]{ .753,  0,  0}{\textbf{93.9 }} & 26.0  & 73.3  & 91.1  \\
     9     & 11.7  & 0.9   & 8.4   & \textcolor[rgb]{ .753,  0,  0}{\textbf{84.4 }} & 44.0  & 72.8  & 74.6  \\
     15    & 25.5  & 29.5  & 21.1  & \textcolor[rgb]{ .753,  0,  0}{\textbf{91.5 }} & 56.3  & 84.4  & 86.9  \\
     \bottomrule
     \end{tabularx}%
   \label{tab:TEP}%
 \end{table*}%
 \begin{figure*}[!htbp]
    \setlength{\abovecaptionskip}{-0.1cm}
	\setlength{\belowcaptionskip}{-0.4cm}
   \begin{center}
   {\includegraphics[width=0.9\textwidth]{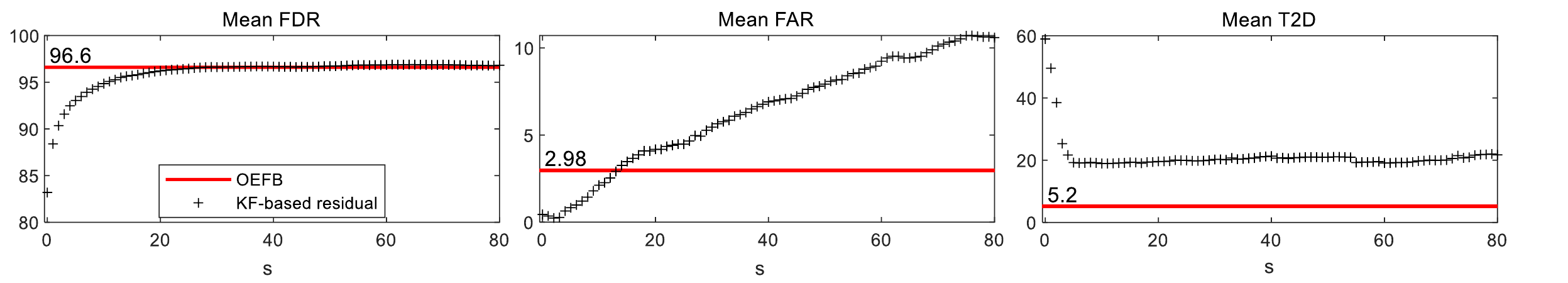}}
   \end{center}
   \caption{Average fault indicators: KF-based methods vs OEFB.}
   \label{fig:TEP}
\end{figure*}

\section{Validation in Tennessee Eastman process}\label{sec:TEP}
TEP is a benchmark process for fault detection and diagnosis. This section will further validate the aforementioned declarations based on the dataset developed in \citep{chiangFaultDetectionDiagnosis2001}. Several model-based fault detection using different identification methods and residuals are compared, including SCVA \citep{luSparseCanonicalVariate2018}, SAP \citep{Ding2009}, OE with $T^2$ statistic, KF-based residual with evaluation function (\ref{eq:J_KF}) and the automatically determined optimal filter scheme proposed in \cite{Zhou2022b}. For KF-based method, different $s$ will be tested. The automatically determined optimal filter scheme uses a filter bank to filter OE, then select the best filter to enhance detection performance, denote it as OEFB. It is an improved version of the optimal filter scheme that does not require fault data while improves all three performance indicators (\ref{eq:index}). 

Concerning the input-output variable selections and model structures used in above methods, please refer the cited papers. SCVA and SAP use subspace identification methods, while others in this paper uses transfer-function prediction method. The KF-based residuals are directly obtained by calculating PE utilizing their equivalence. $H(q)$ in (\ref{eq:pe}) is replaced by its estimate delivered at identification phase. 

The mean FDR, FAR and T2D (time to detection) of 21 documented fault are calculated. Fig.~\ref{fig:TEP} shows the comparison results of them of KF-based methods and OEFB. OEFB is chosen because it gives the best performance in all indicators among state-of-the-art methods \citep{Zhou2022b}. From the figure one can see that as $s$ increases from 0 to 80, the mean FDR is improved, while the mean FAR also increases. Notice that same confidence levels of the thresholds are used. Only when $s>20$, KF-based method gives comparable FDR to OEFB while it gives worse FAR than OEFB. For all tested $s$, mean T2D of KF-based residual is very poor. These observations are same as in Section 4. Further, the FDR of all tested methods of fault 3, 9, 15 are given in Table~\ref{tab:TEP}. These three faults are most difficult to detect \citep{yinComparisonStudyBasic2012}. From the table one can see OEFB that uses optimal filters can give the best FDR while KF-based methods with different $s$ do not.

\section{Conclusion}
\label{sec:con}
This work investigates the optimality of KF-based residual using frequency-domain considerations. Both theoretical analysis and simulation study have been carried out. The conclusion is that Kalman filter is not optimal in general for fault detection. Specifically, the paper reveals: 1) a residual optimal in some index does not necessarily lead a high detection performance. The simulation results show that $\mathcal{J}$ is more relevant to FD performance than $\mathcal{H}_i/\mathcal{H}_{\infty}$. The optimal filter based on $\mathcal{J}$ can offer high FD performance despite the statistical test is not that rigorous; 2) when disturbance and fault are both low-pass, and fault is more low-pass, even with the optimal filter KF-based residual is not better than OE. However there may be other situation in which KF-based residual is better; 3) an optimal statistical test is no doubt desireable in FD, but it does not imply optimality in FD performance; 4) by increasing the window length $s$ in $J_{\RM{KF}}$, the FDR of KF-based residual can be raised. However, a too large $s$ makes little sense in practice because it leads to slow detection speed and requires large data for online calculation.   

There are also occasions Kalman filter being superior. For instance, when working with open-loop unstable system, or developing methods for time-varying system \citep{xu2004residual,zhang2018adaptive,zhong2023fault}. The FD performance of a residual will depend on the system and fault characteristics, as well as the evaluation function and corresponding threshold. There is currently no residual having optimality in all aspects so one must make careful tradeoffs in practice.

\bibliography{mybibfile}

\end{document}